\documentclass[letterpaper, 10 pt, conference]{ieeetran}  
\IEEEoverridecommandlockouts                  
\usepackage{graphicx}          
\usepackage{textcomp}
\usepackage{amsmath}
\usepackage{ifthen}
\usepackage{amssymb}
\usepackage{amsfonts}
\usepackage{subcaption}
\usepackage{colortbl}
\usepackage{algorithm,comment}
\usepackage{algpseudocode}
\usepackage{mathtools,color}
\usepackage{xcolor}
\usepackage{array}
\newtheorem{theorem}{Theorem}
\newtheorem{remark}{Remark}

\title{\LARGE \bf
Sparse Linear Regression with Constraints: A Flexible Entropy-based Framework
}
\author{Amber Srivastava, Alisina Bayati, Srinivasa M. Salapaka
\thanks{\textbf{This work has been submitted to the IEEE for possible publication. Copyright may be transferred without notice, after which this version may no longer be accessible.}}
\thanks{A. Srivastava and A. Bayati contributed equally to this work.}%
\thanks{A. Srivastava is with the Department of Mechanical Engineering, Indian Institute of Technology Delhi, India.
        {\tt\small asrvstv@iitd.ac.in}}%
\thanks{A. Bayati and S.M. Salapaka are with the Coordinated Science Laboratory, University of Illinois Urbana Champaign, USA.
        {\tt\small abayati2/salapaka@illinois.edu}}
}

\begin{document}

\maketitle
\thispagestyle{empty}
\pagestyle{empty}

%%%%%%%%%%%%%%%%%%%%%%%%%%%%%%%%%%%%%%%%%%%%%%%%%%%%%%%%%%%%%%%%%%%%%%%%%%%%%%%%
\begin{abstract}
This work presents a new approach to solve the sparse linear regression problem, i.e., to determine a $k$-sparse vector $\mathbf{w}\in\mathbb{R}^d$ that minimizes the cost $\|\mathbf{y}-A\mathbf{w}\|_2^2$. In contrast to the existing methods, our proposed approach splits this $k$-sparse vector into two parts --- (a) a column stochastic binary matrix $V$, and (b) a vector $\mathbf{x}\in\mathbb{R}^k$. Here, the binary matrix $V$ encodes the location of the $k$ non-zero entries in $\mathbf{w}$. Equivalently, it encodes the  subset of $k$ columns in the matrix $A$ that map $\mathbf{w}$ to $\mathbf{y}$. We demonstrate that this enables modeling several non-trivial application specific structural constraints on $\mathbf{w}$ as constraints on $V$. The vector $\mathbf{x}$ comprises of the actual non-zero values in $\mathbf{w}$. We use Maximum Entropy Principle (MEP) to solve the resulting optimization problem. In particular, we ascribe a probability distribution to the set of all feasible binary matrices $V$, and iteratively determine this distribution and the vector $\mathbf{x}$ such that the associated Shannon entropy gets minimized, and the regression cost attains a pre-specified value. The resulting algorithm employs homotopy from the convex entropy function to the non-convex cost function to avoid poor local minimum. We demonstrate the efficacy and flexibility of our proposed approach in incorporating a variety of practical constraints, that are otherwise difficult to model using the existing benchmark methods.
\end{abstract}

\section{Introduction}\label{sec: Introduction}
Sparse solutions to the linear regression problems have been of interest to multiple fields such as signal and image processing \cite{wen2018survey}, genomics \cite{o2021sparse}, economics and finance \cite{fan2011sparse}, flight load prediction \cite{zhu2019scaled}, machine learning \cite{ChristosL0Sparse}, and remote sensing \cite{tuia2016nonconvex}  
. One of the fundamental ways to formulate this problem is the best subset selection problem, where given a matrix $A=[ a_1~ a_2~\hdots a_d]\in\mathbb{R}^{n\times d}$, a measurement vector $\mathbf{y}\in\mathbb{R}^n$, and a sparsity level $k (\ll d)$, we solve
\begin{align}\label{eq: Sparse_def}
\min_{\mathbf{w}\in\mathbb{R}^d}  \|\mathbf{y}-A\mathbf{w}\|^2_2, \quad \text{subject to } \|\mathbf{w}\|_0 \leq k,
\end{align}
where $\|\mathbf{w}\|_0$ is the number of non-zero entries in $\mathbf{w}$. In other words, the optimization problem (\ref{eq: Sparse_def}) determines (a) the best subset of $k$-columns (features) out of the $d$ columns $\{ a_1,\hdots, a_d\}$ in $A$, and (b) their corresponding coefficients such that the vector $\mathbf{w}$ linearly maps to the measurement vector $\mathbf{y}$ with minimum squared euclidean loss in (\ref{eq: Sparse_def}).

Note that the $\|\cdot\|_0$ norm is non-convex and the optimization problem (\ref{eq: Sparse_def}) is NP-hard \cite{welch1982algorithmic}. A large number of work done in this area design approximate solutions to (\ref{eq: Sparse_def}), where they iteratively add or remove the non-zero coefficients in $w$ to minimze (\ref{eq: Sparse_def}); for instance matching pursuit \cite{zhang2011sparse,donoho2012sparse} and forward-backward approaches \cite{zhang2008adaptive,shekhar2022forward}. See \cite{zhang2015survey} for a detailed survey. Other relevant approaches involve replacing the non-convex constraint in (\ref{eq: Sparse_def}) with a sparsity promoting term $\mathcal{T}(\mathbf{w})$ in the objective function, and solving the optimization problem
\begin{align}\label{eq: Reg_version}
\min_{\mathbf{w}\in\mathbb{R}^d} \|\mathbf{y}-A\mathbf{w}\|_2^2 + \lambda \mathcal{T}(\mathbf{w}),
\end{align}
where $\lambda$ is a regularization parameter. A popular choice for $\mathcal{T}(\mathbf{w})$ is the $l_1$ norm $\|\mathbf{w}\|_1$, which results into a convex optimization problem (\ref{eq: Reg_version}). Algorithms such as gradient projection \cite{figueiredo2007gradient}, iterative shrinkage-thresholding \cite{beck2009fast}, and (linear) alternating direction method \cite{yang2013linearized} guarantee globally optimal solutions, and several heuristics such as orthogonal matching pursuit \cite{davis1997adaptive} and least angle regression \cite{LARS} efficiently address this convex program.

Another class of choice for $\mathcal{T}(\mathbf{w})$ that has received much attention lately are the non-convex regularizers, which have been shown to result into better solutions than their convex counterpart \cite{amir2021trimmed}. Some of these choices are $\|\mathbf{w}\|_p$ with $0<p<1$ \cite{wang2021non}, minimax concave penalty (MCP) \cite{9585422}, smoothly clipped absolute deviation (SCAD) \cite{SCAD}, and trimmed lasso \cite{amir2021trimmed,bertsimas2017trimmed}. The latter has the additional property that it {\em exactly} results into a sparsity level $k$ as indicated in the optimization problem (\ref{eq: Sparse_def}). 

As is evident from above, the work done in this area is extensive, with several different proposed frameworks that address various aspects of the problem such as scalability, computational costs, bias and exactness of sparsity. See \cite{zhang2015survey} for a survey on these methods. Various scenarios such as (overlapping) grouped variables \cite{wang2008note,yuan2011efficient}, shape constraints \cite{micchelli2013regularizers} and restricted non-zero values \cite{baldassarre2012incorporating} impose additional constraints on the design of the sparse vector $\mathbf{w}$ in (\ref{eq: Sparse_def}). Though there are methods to address such specific structural constraints, there is, to the best of our understanding, limited work on a generalized framework that effectively {\em models} and incorporates such constraints in (\ref{eq: Sparse_def}).  

The sparsity and structural constraints on $\mathbf{w}$ can alternatively be viewed as constraints on the selection of the feature vectors $\{ a_j\}_{j=1}^d$ from the matrix $A$. Thus, a {\em direct} control over the selection of these feature vectors will provide flexibility in modeling a variety of structural constraints discussed above, and also in enforcing the sparsity level $k$ of the vector $\mathbf{w}$ (which, generally speaking, is also a structural constraint). To this end, we develop a framework that (a) provides a direct control over the desired level of sparsity in the vector $\mathbf{w}$, (b) is flexible to incorporate a wide-range of application specific structural constraints on $\mathbf{w}$, and (c) results into an algorithm that is designed to avoid poor local minima of the underlying non-convex optimization problem.

The above contribution (a) result from our viewpoint of the optimization problem (\ref{eq: Sparse_def}), where we dissociate the $k$-sparse vector $\mathbf{w}\in \mathbb{R}^d$ into two parts --- a binary column stochastic matrix $V\in\{0,1\}^{d\times k}$ and a vector $\mathbf{x}\in\mathbb{R}^k$. The matrix $V$ is designed to encode the location of the $k$ non-zero values in $\mathbf{w}$. Equivalently, it directly controls subset of $k$ columns (features) in the design matrix $A$, that map the non-zero entries in $\mathbf{w}$ to the measurement $\mathbf{y}$. The vector $\mathbf{x}\in\mathbb{R}^k$ comprises of these non-zero values in $\mathbf{w}\in\mathbb{R}^d$. As elaborated in the Section \ref{sec: PF_sparsity}, the column stochasticity and the size of the binary matrix $V$, and the size of the vector $\mathbf{x}$ guarantee that the desired level $k$ of sparsity in $\mathbf{w}$ is {\em exactly} achieved.

The contribution (b) also results from the decision matrix $V$. Since the matrix $V$ governs the choice of the $k$ features in the design matrix $A$, it explicitly enables modeling several structural constraints that restrict the {\em permissible choice of subsets of $k$ features} in $A$. For example, (as demonstrated later) constraints such as selecting only 2 out of the 4 given features, not allowing all features in a given subset $\{ a_{l_1}, a_{l_2}, a_{l_3}\}$ to be selected, or modeling existing constraints such as selecting pre-defined groups of features (popularly addressed using group lasso \cite{jacob2009group}) can be conveniently modelled as structural constraints on $V$. As far as we are aware, our proposed framework is the most flexible in incorporating such variety of constraints on the permitted choice of the features; primarily owing to the introduction of the matrix parameter $V$ in our model that explicitly determines the choice of the features in $A$.

The contribution (c) results from the use of Maximum Entropy Principle (MEP) in determining the matrix $V$ and the vector $\mathbf{x}$. Note that the matrix $V$ is a discrete decision variable that lies in a {\em combinatorially} large set $\mathcal{V}$ of all possible binary column stochastic matrices. Thus, the sparse linear regression (SLR) problem, with $V$ and $\mathbf{x}$ as the decision variable, can be viewed as a combinatorial optimization problem. In the past, MEP-based frameworks have successfully addressed a variety of such problems; for instance the facility location problem \cite{rose1998deterministic}, data aggregation \cite{xu2014aggregation}, vehicle routing \cite{baranwal2022unified}, network design \cite{9517030}, protein structure alignment \cite{chen2005protein}, and image processing \cite{yu2013maximal}. The abstract idea behind all these frameworks is to consider a distribution over the set of all possible values of the discrete variable. Then, determine the distribution that maximizes the Shannon entropy \cite{jaynes2003probability} at a pre-specified value of the  expected cost function. This results into an iterative process, wherein the pre-specified value is successively lowered to as small value as possible and the solution from the previous iteration forms an initialization for the next. These iterations mimic a homotopy from the convex entropy to the non-convex cost function, which prevents the algorithm from getting stuck in a poor local minima \cite{9517030}. As described later in Section \ref{sec: PS_sparsity}, instead of considering the distribution over the combinatorially large set $\mathcal{V}$, we introduce auxiliary distributions over the individual entries $v_{ij}$ in $V$ (at the cost of an additional constraint); thus, making the resulting optimization problem computationally tractable.

We observe that the proposed MEP-based framework performs as good as the recent trimmed lasso method on the unconstrained optimization problems \cite{amir2021trimmed}, and outperforms the convex regularization based methods such as lasso, ridge regression, LARS, and adaptive lasso. We demonstrate the frameworks flexibility in handling various practical constraints (as discussed above). We also illustrate and analyze the characteristic features of the MEP-based framework such as annealing and the phase transitions, and their utility towards increasing computational efficiency and determining the choice of sparsity level $k$ in the SLR (\ref{eq: Sparse_def}).

\begin{comment}
\section{Maximum Entropy Principle}\label{sec: MEP_sparsity}
As Maximum Entropy Principle (MEP) \cite{jaynes1957information} is significant to our proposed approach, we provide a brief description below. Let $\mathcal{X}$ be a random variable whose realization lie in the set $\{x_1,\hdots,x_n\}$. MEP states that given a priori information about $\mathcal{X}$, the most unbiased probability distribution over it is the one that maximizes the Shannon entropy. More precisely, let this information be given in the form of the constraint on the expected value of the function $f_k: \mathcal{X} \rightarrow \mathbb{R}$, i.e.,
\begin{align}\label{eq: MEP_inf}
\sum_{i = 1}^n p_{\mathcal{X}} (x_i) f_k(x_i) = F_k\quad\forall 1 \leq k \leq m,
\end{align}
where $F_k$ are known. Then, as per MEP, the most unbiased probability distribution over $\mathcal{X}$ solves
\begin{align}\label{eq: MEP_optimization}
\begin{split}
    \max_{\{p_{\mathcal{X}}(x_i)\}} \quad &H(\mathcal{X}) = - p_{\mathcal{X}}(x_i) \log (p_{\mathcal{X}}(x_i))\\
    \text{subject to: } &\sum_{i = 1}^n p_{\mathcal{X}} (x_i) f_k(x_i) = F_k \quad \forall \; 1 \leq k \leq m,
\end{split}
\end{align}
which results into the distribution
\begin{align}\label{eq: MEP_Gibbs}
p_{\mathcal{X}} (x_i) = \frac{e^{- \sum_{k=1}^m \lambda_k f_k(x_i)}}{\sum_{j = 1}^n e^{- \sum_{k=1}^m \lambda_k f_k(x_j)}},
\end{align}
where $\lambda_k$, $1 \leq k \leq m$ denote the corresponding Lagrange multipliers of the constraints in (\ref{eq: MEP_optimization}).
\end{comment}

\section{Maximum Entropy Principle}\label{sec: MEP_overview}
The Maximum Entropy Principle (MEP), introduced by Jaynes \cite{jaynes1957information}, plays a pivotal role in shaping our proposed methodology, and we present a brief elucidation herein. Let us consider a random variable $\mathcal{X}$ characterized by realizations within the set $\{x_1, \ldots, x_n\}$. MEP states that, when armed with a priori information regarding $\mathcal{X}$, the most unbiased probability distribution is one that maximizes the Shannon entropy.
More precisely, if this a priori knowledge is articulated as constraints on the expected values of functions $f_k: \mathcal{X} \rightarrow \mathbb{R}$,
\begin{align}\label{eq: MEP_inf}
\sum_{i = 1}^n p_{\mathcal{X}} (x_i) f_k(x_i) = F_k \quad \forall 1 \leq k \leq m,
\end{align}
where $F_k$ denotes known constants, MEP dictates that the most unbiased probability distribution for $\mathcal{X}$ is obtained by solving the optimization problem:
\begin{align}\label{eq: MEP_optimization}
\begin{split}
    \max_{\{p_{\mathcal{X}}(x_i)\}} \quad &\mathcal{H}(\mathcal{X}) = - p_{\mathcal{X}}(x_i) \log (p_{\mathcal{X}}(x_i))\\
    \text{subject to: } &\sum_{i = 1}^n p_{\mathcal{X}} (x_i) f_k(x_i) = F_k \quad \forall \; 1 \leq k \leq m,
\end{split}
\end{align}
The optimization process leads to the determination of the Gibbs distribution as the optimal probability distribution, striking a balance between unbiasedness and adherence to the apriori information.

\section{Problem Formulation}\label{sec: PF_sparsity}
As briefly stated in the Section \ref{sec: Introduction}, we begin by re-writing the sparse vector $\mathbf{w}$ as a product of a matrix $V$ and a vector $\mathbf{x}$, i.e. $\mathbf{w}=V\mathbf{x}$, where the matrix $V$ lies in the set 
\begin{align}\label{eq: SetV}
\mathcal{V}:=\{V=(v_{ij})\in\{0,1\}^{d\times k}: \sum_i v_{ij}=1~\forall~j\}
\end{align} 
and $x\in\mathbb{R}^k$. Note that the number of columns in the matrix $V$, the size of the vector $\mathbf{x}$, binary entries in $V$, and column-stochasticity of $V$ (i.e., $\sum_{i}v_{ij}=1~\forall~j$) ensure that maximum number of non-zero elements in $\mathbf{w}=V\mathbf{x}$ are exactly $k$. For instance, let $V\in\mathcal{V}$ be such that $v_{rs}=1$, then the $r$-th position in $\mathbf{w}$ is non-zero and is taken up by the $s$-th entry of $x$. Further, the column stochasticity of $V$ ensures that the $s$-th entry of $\mathbf{x}\in\mathbb{R}^k$ does not appear at any other location in $\mathbf{w}$ --- thereby, guaranteeing $k$ non-zero values in $\mathbf{w}$. The fact that the vector $\mathbf{x}$ lies in $\mathbb{R}^k$, and that the matrix $V$ lies in the set $\mathcal{V}$ in (\ref{eq: SetV}) together are equivalent to the sparsity constraint $\|\mathbf{w}\|_0\leq k$. Thus, we re-write the sparse linear regression (SLR) problem in (\ref{eq: Sparse_def}) using the above notations as
\begin{align} \label{eq: Sparse_def2}
    \min_{\mathbf{x}\in\mathbb{R}^k, V\in\mathcal{V}}  \|\mathbf{y}-A V \mathbf{x}\|^2_2.
\end{align} 

\section{MEP-based framework for Problem Solution}\label{sec: PS_sparsity}
To make the optimization problem in (\ref{eq: Sparse_def2}) amenable to be addressed within an MEP-based framework, we reformulate it as
\begin{subequations}\label{eq: Sparse_def3}
\begin{align}
\min_{\substack{\mathbf{x}\in\mathbb{R}^k,\\ \{\eta(V|\mathbf{x})\}}}&~\sum_{V\in\mathcal{V}}\eta(V|\mathbf{x}) \|\mathbf{y}-A V \mathbf{x}\|^2_2\\
\text{subject to}& ~\eta(V|\mathbf{x})\in\{0,1\}, \text{ and } \sum_{V\in\mathcal{V}}\eta(V|\mathbf{x}) = 1,
\end{align}
\end{subequations}
where $\eta(V|\mathbf{x})$ is an auxiliary binary decision variable that determines the matrix $V$. We then replace $\eta(V|\mathbf{x})\in\{0,1\}$ by the soft decision variable $p(V|\mathbf{x})\in[0,1]$, resulting into a relaxed regression cost
\begin{align}\label{eq: RelaxCost}
D:=\sum_{V\in\mathcal{V}}p(V|\mathbf{x})\|\mathbf{y}-A V\mathbf{x}\|_2^2.
\end{align}
Note that $p(\cdot|\mathbf{x})$ can also be interpreted as the discrete distribution over the space of all the matrices $V\in\mathcal{V}$ given $\mathbf{x}$. We use MEP to design this distribution $\{p(V|\mathbf{x})\}$ as well as to determine the vector $\mathbf{x}$. In particular, their design is based on the principle of maximizing the Shannon entropy $H$ subject to the constraint that the expected cost function $D$ in (\ref{eq: RelaxCost}) attains a pre-determined value $c_0$. This leads to the following optimization problem.
\begin{subequations}\label{eq: MEP_sparse}
\begin{align} \label{eq: MEP_sparse_obj}
    \max_{\substack{\{p(V|\mathbf{x})\}\\ \mathbf{x}\in\mathbb{R}^k}}H&:= -\sum_{V}p(V|\mathbf{x})\log(p(V|\mathbf{x}))\\ 
    \text{subject to }& D:=\sum_{V}p(V|\mathbf{x})\|\mathbf{y}-A V\mathbf{x}\|_2^2 = c_0, \label{eq: MEP_sparse_c1}\\ 
    & \sum_{V}p(V|\mathbf{x}) = 1, ~V\in\mathcal{V}.\label{eq: MEP_sparse_c2} 
\end{align}
\end{subequations}
Since ${|\mathcal{V}|=d^k}$, the resulting decision variable space $\{p(V|\mathbf{x})\}$ is exponentially large; thus, making the optimization problem (\ref{eq: MEP_sparse}) intractable in its current form. We trim down the decision variable space to polynomial order by dissociating the decision variable $p(V|\mathbf{x})$ as  
\begin{align}\label{eq: eta_break}
p(V|\mathbf{x}) = \prod_{i,j=1}^{d,k}p_{ij}(v_{ij}|\mathbf{x}),
\end{align}
where $p_{ij}(\cdot|\mathbf{x})$ is distribution over all possible values $v_{ij}\in\{0,1\}$, the $ij$-th entry in $V$, takes. The new decision variable space $\{\{p_{ij}(v_{ij}|\mathbf{x})\},x\}$ is now of the polynomial order $\mathcal{O}(dk)$, which takes us closer to posing the optimization problem (\ref{eq: MEP_sparse}) in a computationally  tractable way. Substituting (\ref{eq: eta_break}) in the objective (\ref{eq: MEP_sparse_obj}) we obtain
\begin{equation}
\begin{aligned}\label{eq: Simplif_obj}
    \mathcal{H} := \mathbf{1}_d^{\top}[Q\circ\log Q+\bar{Q}\circ\log \bar{Q}]\mathbf{1}_k,
\end{aligned}
\end{equation}
where $\circ$ denotes element wise operation, $Q\in[0,1]^{d\times k}$ and $\bar Q\in[0,1]^{d\times k}$ are defined as
\begin{equation}
\begin{aligned}\label{eq: Q}
    Q:=(q_{ij}), \; q_{ij}:=p_{ij}(v_{ij}=1|\mathbf{x}), \; \bar{Q}=\mathbf{1}_{d\times k}-Q.
\end{aligned}
\end{equation}
Please see Appendix \ref{app: Lin_obj_simplif} for details on the above (as well as the following) algebraic simplifications. The constraint (\ref{eq: MEP_sparse_c1}), in terms of $Q$ and $\bar{Q}$, transforms into
\begin{equation}
    \begin{aligned}\label{eq: Simplif_const}
    \mathcal{D}:=&\|\mathbf{y}-A Q\mathbf{x}\|^2_2\\
    &+ [ a_1^{\top} a_1~\hdots~  a_d^{\top} a_d][Q\circ\bar{Q}](\mathbf{x}\circ\mathbf{x})=c_0,
    \end{aligned}
\end{equation}
and the constraint (\ref{eq: MEP_sparse_c2}), which ensures that only one $V$ in selected from the set $\mathcal{V}$ and that $V$ is a column stochastic matrix, is taken care by the fact that we define $\bar{Q}:=(\bar{q}_{ij})$ (i.e., $p_{ij}(v_{ij}=0|x)$) as $1-Q$ (where $q_{ij}=p_{ij}(v_{ij}=1|x)$), and that $Q^{\top}\mathbf{1}_d =\mathbf{1}_k$, i.e., $Q$ is also column stochastic matrix. More precisely, the reformulation of the optimization problem (\ref{eq: MEP_sparse}) in terms of the tractable decision variables is
\begin{equation}
    \begin{aligned} \label{eq: MEP_sparse_Lin}
    \max_{\substack{\mathbf{x}\in\mathbb{R}^k, Q\in[0,1]^{d\times k}}} &\mathcal{H} \\
    \text{subject to }& \mathcal{D} = c_0, \quad Q^{\top}\mathbf{1}_d =\mathbf{1}_k.
    \end{aligned}
\end{equation} 
We consider the following augmented Lagrangian corresponding to the optimization problem (\ref{eq: MEP_sparse_Lin})
\begin{equation}
    \begin{aligned} \label{eq: Lag2}
    F_{T} =& \mathcal{H} - \frac{1}{T}(\mathcal{D}-c_0) - {\mu}^{\top} (Q^{\top}\mathbf{1}_d - \mathbf{1}_k)\\
    &\qquad- \frac{1}{2}\rho\|Q^{\top}\mathbf{1}_d-\mathbf{1}_k\|^2_2,
    \end{aligned}
\end{equation}
where $T$ and $\mu$ denote the Lagrange multipliers corresponding to the constraints in (\ref{eq: MEP_sparse_Lin}), $\|Q^{\top}\mathbf{1}_d-\mathbf{1}_k\|^2_2$ denotes the penalty term, and $\rho$ denotes the penalty parameter. Due to its close analogy to the MEP-based framework illustrated in \cite{rose1991deterministic}, we refer to $T$ as the temperature, and $F_T$ as the free-energy term. It is known from sensitivity analysis \cite{jaynes2003probability} that a large value of the Lagrange parameter $T$ corresponds to a large value of $c_0$. Similarly, a small value of $T$ corresponds to a small value of $c_0$. In our framework, we repeatedly solve (\ref{eq: MEP_sparse_Lin}) at decreasing values of $c_0$ by maximizing the Lagrangian $F_T$ at iteratively decreasing values of $T$. 

More precisely, let $T_k$ be the temperature value at the $k$-th iteration of the algorithm, the penalty parameter $\rho_k$ be equal to $T_k$, and the multiplier $\mu_k$ be given by the iteration 
\begin{align}\label{eq: MuIterate}
\mu_k=\mu_{k-1}+\rho_k(Q^{\top}\mathbf{1}_d-\mathbf{1}_k),
\end{align}
where $\mu_0=0$. We vary $T_k$ from a large value $(\rightarrow\infty)$ to a small value $(\approx 0)$. At large values of $T_k$, the Lagrangian $F_{T_k}$ is dominated by the convex entropy function $\mathcal{H}$ and the penalty parameter. As $T_k$ becomes small, the other terms including the non-convex $\mathcal{D}$ gets more weightage. As in other MEP-based frameworks, it is this homotopy from a convex function to the non-convex cost function $\mathcal{D}$ that helps the algorithm avoid getting stuck in a poor local minima. It is also known that given $T_{k+1}\geq T_k$, the above iterations converge to a local minima of the optimization problem (\ref{eq: MEP_sparse_Lin}) (see \cite{bertsekas2014constrained,hestenes1969multiplier,powell1969method} for details). Please see Algorithm \ref{alg: algorithm_SLR} for details on implementation.
\begin{algorithm}[t]
\caption{Maximum Entropy Sparsity-enforcing Regularization for Linear Regression}\label{alg: algorithm_SLR}
\begin{algorithmic}[1]
\State{\textbf{Input: }} $T_{\min}$, $T_{\max}$, $k$, $\beta<1$;  \State{\textbf{Output: }} $Q$ and $\mathbf{x}$.
\State{\textbf{Initialize:}} $T_k= T_{\max}$, $Q_0=[q_1,\hdots,q_k]$, $q_j=\frac{1}{d}  \mathbf{1}_d$, $k=1$, $\rho_0=\frac{1}{T_0}$.
\While{$T_k\geq T_{\min}$}
\State{\textbf{Obtain $Q_k$, $x_k$:}} Minimize $F_T$ in (\ref{eq: Lag2}) using a descent method and initial value $Q_{k-1}$, $x_{k-1}$.
\State $T_k \leftarrow \beta T_k$
\State $k\leftarrow k+1$.
\State Update $\mu_k$ in (\ref{eq: MuIterate}) and $\rho_k=\frac{1}{T_k}$.
\EndWhile
\end{algorithmic}
\end{algorithm}

{Convergence of $Q$ to a binary matrix :} As our objective is to solve the optimization problem (\ref{eq: Sparse_def3}), we want that as $T\rightarrow 0$, the matrix $Q$ converges to a binary matrix, i.e., $Q\rightarrow\{0,1\}^{d\times k}$. This in turn enables $p(V|\mathbf{x})\rightarrow \{0,1\}$, i.e., the soft decision variables converge to the binary solution required in (\ref{eq: Sparse_def3}). The structure of the optimal $Q$ is amenable to this desired aspect of the Algorithm \ref{alg: algorithm_SLR}. More precisely, by setting $\frac{\partial F_T}{\partial \mathbf{x}}=0$ and $\frac{\partial F_T}{\partial Q}=0$, we obtain
\begin{subequations}\label{eq: Sol_SLR}
\begin{align} \label{eq: x_SLR}
&\mathbf{x}=\Big[Q^{\top}\;A ^{\top}A \; Q + \text{diag}\big[\lambda^{\top}(Q\circ\bar{Q})\big]\Big]^{-1}Q^{\top}A ^{\top}\mathbf{y},\\
&Q=\frac{\exp(\frac{2}{T}H_m)\circ}{\exp(\frac{2}{T}H_m)+\exp(-\frac{2}{T}H_p)}, \label{eq: Q_SLR}
\end{align}
\end{subequations}
where $\circ$ denotes elementwise operation, 
\begin{subequations}
\begin{align}
&H_m=\min\Big\{\Xi,\mathbf{0}\Big\}, \quad H_p=\max\Big\{\Xi,\mathbf{0}\Big\}, \\
&\Xi = A^{\top}(\mathbf{y}-A Q\mathbf{x})\mathbf{x}^{\top} - \frac{1}{2}\lambda_a(\mathbf{x}\circ\mathbf{x})^{\top}\circ(1-2Q)\\ 
&\qquad -\frac{1}{2}\mathbf{1}_d\mu^{\top}-\frac{1}{2T}\mathbf{1}_d\mathbf{1}_d^{\top}Q+\frac{1}{2T}\mathbf{1}_d\mathbf{1}_k^{\top},
\end{align}
\end{subequations}
and $\lambda_a=[ a_1^{\top} a_1~~ \cdots ~~ a_d^{\top} a_d]$. Note that $Q$ in (\ref{eq: Q_SLR}) resembles a  Gibb's distribution, whose entries are identical at large values of $T (\rightarrow\infty)$, and converge to either $0$ or $1$ as $T\rightarrow 0$ (with $\Xi$ bounded); thus, achieving the above objective.

\begin{comment}
{\color{blue}In fact, we solve the optimization problem (\ref{eq: Lag_final}) repeatedly at decreasing values of $T$. At large values of $T\rightarrow \infty$, the Lagrangian $F$ is dominated by the convex function $-H$, the weights $Q$ in (\ref{eq: Q_SLR}) are uniformly distributed ($q_{ij}=\frac{1}{d}$), and all the entries of $\mathbf{x}$ are identical; i.e., the algorithm effectively finds just {\em one} distinct entry in $\mathbf{x}$, or equivalently, the solution is of sparsity $k=1$ at large values of temperature $T$. As $T$ decreases further the Lagrangian $F$ is more and more dominated by the cost function $D$ in (\ref{eq: F_define}), the weights $q_{ij}$ in $Q$ are no longer uniform, and the elements of $\mathbf{x}$ become distinct from each other, i.e., the sparsity $k>1$. At $T\rightarrow 0$, the weights in $Q$ converge to either $0$ or $1$, the Lagrangian $F$ converges to $D$, and the algorithm solves the original optimization problem in (\ref{eq: Sparse_def3}).

In the upcoming section, we will elaborate on the gradual increase in the number of unique elements in the variable $\mathbf{x}$. We will also establish analytical conditions to further support our argument.}
\end{comment}

\section{Flexibility of the Framework in Modeling Constraints}\label{sec: Flexibility}

As briefly discussed in the Section \ref{sec: Introduction}, our proposed framework explicitly allows the control over the selection of the features in the design matrix $A$. More precisely, the $j$-th feature $ a_j$ is selected if and only if the sum of $j$-th row in the binary matrix $V$ is non-zero. More precisely, if the $j$-th feature is selected, then $\sum_{t=1}^k v_{jt} \geq 1$.
\begin{comment}
\begin{align}
\sum_{t=1}^k v_{jt} \geq 1
\end{align}     
\end{comment}
This attribute of the framework allows us to conveniently model several structural constraints in the design of the sparse vector $\mathbf{w}$, which otherwise are difficult to model in the existing literature when explicit control over the selection of the feature is not possible. Below we elucidate some of these scenarios.

{\em 1- Correlated feature vectors: } The columns of a given design matrix $A$ may be linearly dependent on each other, where the extent of their linear dependence is measured by the Pearson correlation coefficient \cite{cohen2009pearson}. For highly correlated set of features, it is desirable to have only one of the features to be selected, i.e., to have a non-zero entry in the sparse vector $\mathbf{w}$ only for one of such features, and to have zero values in $\mathbf{w}$ corresponding to all other features in this set. One straightforward solution is to drop all the features in this set, except the one that highly correlates with the output $\mathbf{y}$. However, such a methodology is sub-optimal. On the other hand, our proposed framework explicitly models this constraint in terms of the matrix $V$. In particular, let the $r$ features in the set $\{ a_{l_1}, a_{l_2},\hdots, a_{l_r}\}$ be highly correlated. Then the constraint $\sum_{t=1}^m v_{l_1t} + v_{l_2t} + \hdots +v_{l_rt} \leq 1,$
\begin{comment}
\begin{align}\label{eq: Const1Correlation}
\sum_{t=1}^m v_{l_1t} + v_{l_2t} + \hdots +v_{l_rt} \leq 1,
\end{align}    
\end{comment}
enforces that at most one of the above $r$-features is picked. To incorporate this in the optimization problem (\ref{eq: MEP_sparse_Lin}), we begin by re-writing it as
\begin{align}\label{eq: Const1Correlation2}
\sum_{V\in\mathcal{V}}\eta(V|x)\big(\sum_{t=1}^k v_{l_1t} + v_{l_2t} + \hdots +v_{l_rt}\big)\leq 1.
\end{align}
As done before, we replace the above $\eta(V|x)\in\{0,1\}$ with soft weights $p(V|x)\in[0,1]$, and dissociate them as in (\ref{eq: eta_break}). Subsequently, the algebraic manipulations (similar to that in (\ref{eq: Simplif_obj}) and (\ref{eq: Simplif_const})) result into a constraint in terms of the decision variable $Q=(q_{ij})$ as follows
\begin{align}\label{eq: Const1Correlation3}
\sum_{t=1}^k q_{l_1t} + q_{l_2t} + \hdots +q_{l_rt} \leq 1.
\end{align}
The above constraint can be incorporated easily into the optimization problem (\ref{eq: MEP_sparse_Lin}).

{\em 2- A priori knowledge: } Expert insights  play a substantial role in shaping the sparse solution $\mathbf{w}$, often suggesting the inclusion of at least one feature from each group when working with multiple groups of features. For instance, in the context of medical diagnosis, it is recommended to utilize data from distinct diagnostic groups, such as radiological imaging, clinical laboratory tests, and patient medical history, to ensure a comprehensive evaluation of a patient's condition. Our framework allows us to incorporate such a priori information into the existing problem. To elaborate, suppose there are $r$ features in the set $\{ a_{l_1}, a_{l_2},\ldots, a_{l_r}\}$ originating from the same group of features. In this context, the constraint $\sum_{t=1}^k v_{l_1t} + v_{l_2t} + \ldots +v_{l_rt} \geq 1$ serves to ensure that at least one of the $r$ features is selected. Similar to the previous scenario, it can be expressed in terms of the decision variable $Q$ as follows
\begin{align}\label{eq: Const2Correlation2}
\sum_{t=1}^k q_{l_1t} + q_{l_2t} + \hdots +q_{l_rt} \geq 1,
\end{align}

{\em 3- Grouping constraints: } In many applications, a group of features needs to be selected as a single unit. Algorithms such as group lasso addresses such instances by introducing a regularization term $\mathcal{T}(\mathbf{x})$ to the cost function in (\ref{eq: Sparse_def}). In the proposed framework, such group constraints can be easily modeled. For instance, let $\{ a_{l_1}, a_{l_2},\hdots, a_{l_r}\}$ be a group of features that occurs as a single unit, i.e., if one of them is picked then all of them should be picked. This constraint is modeled as
\begin{align}\label{eq: Const3Correlation}
v_{l_1t}=v_{l_2t}=\cdots=v_{l_rt}
\end{align}
As above, the constraint (\ref{eq: Const3Correlation}) in terms of the decision variable $Q$ is given by $q_{l_1t}=q_{l_2t}=\cdots=q_{l_rt}$.

\begin{comment}
    Similarly, nonlinear constraints in terms of the binary decision matrix $V$ (equivalently, $Q$) can be developed to incorporate much larger class of constraints such as the overlapping group constraint \cite{yuan2011efficient}. For example, let $\{a_{l_1},a_{l_2},a_{l_3}\}$ and $\{a_{l_3},a_{l_4},a_{l_5}\}$ be two overlapping sets of grouped features. The following constraints enforces that if a feature is chosen from one of the set then the entire set is selected
{\color{red}\begin{align}\label{eq: Const4Correlation}
&\Bigg(\sum_{i,t=1}^{3,k}\frac{q_{l_it}}{3}\Bigg)\Bigg(1-\sum_{i,t=1}^{3,k}\frac{q_{l_it}}{3}\Bigg)\nonumber\\
&\qquad + \Bigg(\sum_{i=3,t=1}^{5,k}\frac{q_{l_it}}{3}\Bigg)\Bigg(1-\sum_{i=3,t=1}^{5,k}\frac{q_{l_it}}{3}\Bigg)=0
\end{align}}
Note that if both the sets are partially selected then the above expression is non-zero. 
\end{comment}

Note that the above inequality and equality based constraints can be easily incorporated into (\ref{eq: MEP_sparse_Lin}), and the resulting optimization problem can be addressed using existing methods such as interior points algorithm \cite{byrd1999interior, coleman1996interior} and various other penalty methods \cite{bertsekas2014constrained}.

\begin{figure}
    \centering
    \includegraphics[width=0.9\columnwidth]{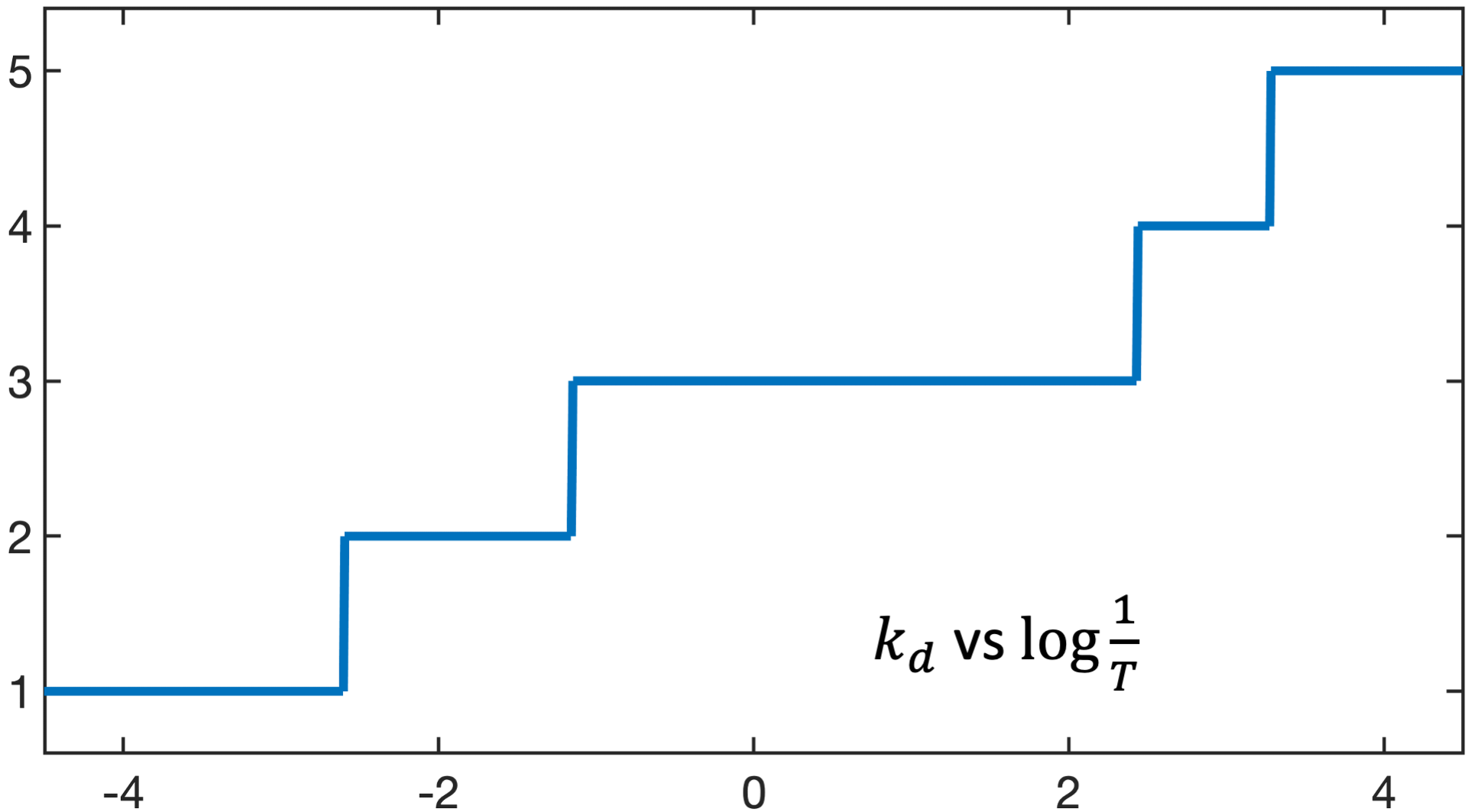}
    \caption{Phase Transition. Plot of $k_d$ versus $T$, where $k_d$ are the distinct number of non-zero values in $\mathbf{x}\in\mathbb{R}^k$ (or, equivalently the distinct columns in $Q\in[0,1]^{d\times k}$).}
    \label{fig: PT_figure}
\end{figure} 
\section{Phase Transition}\label{sec: PT_sparsity}
Algorithm \ref{alg: algorithm_SLR} is characterized by a unique trait wherein at large values of temperatures $T$ all non-zero values in $\mathbf{x}$ are {\em identical}; equivalently, the columns in $Q$ are identical. As $T$ decreases, there are specific instances at which the number of distinct non-zero values in $\mathbf{x}$ increase (equivalently, distinct columns in $Q$ increase). We refer to these instances as {\em phase transitions}, and the associated temperature values at which they occur as critical temperatures $T_{cr}$. Figure \ref{fig: PT_figure} illustrates this phenomenon on randomly generated data $\mathbf{y}\in\mathbb{R}^8$, $A\in\mathbb{R}^{8\times 15}$, $\mathbf{w}\in\mathbb{R}^{15}$, and $\|w\|_0=3$. For the purpose of illustration we set $k=5$ (though the {\em true} sparsity is $3$, it is not known a priori in general). The Algorithm \ref{alg: algorithm_SLR} begins with $k_d=1$ distinct non-zero value in $x\in\mathbb{R}^5$ at high temperatures (i.e., low values of $\log(1/T)$). As $T$ decreases, the number of distinct values remain unchanged for sometime before a critical temperature is reached, where $\mathbf{x}\in\mathbb{R}^5$ contains $k_d=2$ distinct values. This process continues, till the Algorithm \ref{alg: algorithm_SLR} determines $k_d=5$ distinct non-zero-values in $\mathbf{x}\in\mathbb{R}^5$ (equivalently, all columns in $Q$ are distinct).
 
It is not uncommon for MEP-based frameworks to exhibit such characteristics. Phase transitions have been observed in the various contexts such as data clustering \cite{rose1998deterministic}, Markov chain aggregation \cite{xu2014aggregation}, and facility location and path optimization \cite{srivastava2020simultaneous}. They have been shown to occur when the stationary point of the Lagrangian is no longer a minima. More precisely, let $F_T^*(Q)$ be the Lagrangian obtained after substituting $\mathbf{x}$ (\ref{eq: x_SLR}) in augmented Lagrangian $F_T$ in (\ref{eq: Lag2}). Then, in the current context phase transition occurs when the stationary point of the Lagrangian $F_T^*(Q)$ becomes a saddle point. This observation is crucial to explicitly quantifying the critical temperature $T_{cr}$ as instances where the Hessian $\mathcal{J}$ corresponding to the Lagrangian $F_T^*$ loses its positive definite property (i.e. local minimum becomes a saddle point). We use variational calculus to determine $T_{cr}$. Let $Q^*$ be the optimal to the Lagrangian $F_T^*$, and $\Psi = [\psi_1~\psi_2~\hdots~\psi_d]^{\top}\in\mathbb{R}^{d\times k}$ denote a feasible perturbation direction along which the Hessian $\mathcal{J}$ is not positive, i.e., $\mathcal{J}(Q,\Psi,T):= $  
\begin{equation}
\begin{aligned} \label{eq:  HessSLR}
\frac{d^2 F^*(Q+\epsilon\Psi)}{d\epsilon^2}\Big|_{\epsilon=0} &= \sum_{j=1}^k\Psi_j^{\top}\Big[\frac{T}{2}\Upsilon_{j}-\Gamma_j\Big]\Psi_j \\
    &\quad+\Big\|\sum_{j=1}^k E_jC\Psi_j\Big\|^2 \leq 0,
\end{aligned} 
\end{equation}
where the matrices in (\ref{eq: HessSLR}) are defined in the Appendix \ref{app: matrices}.
\begin{figure}
    \centering
    \includegraphics[width=0.95\columnwidth]{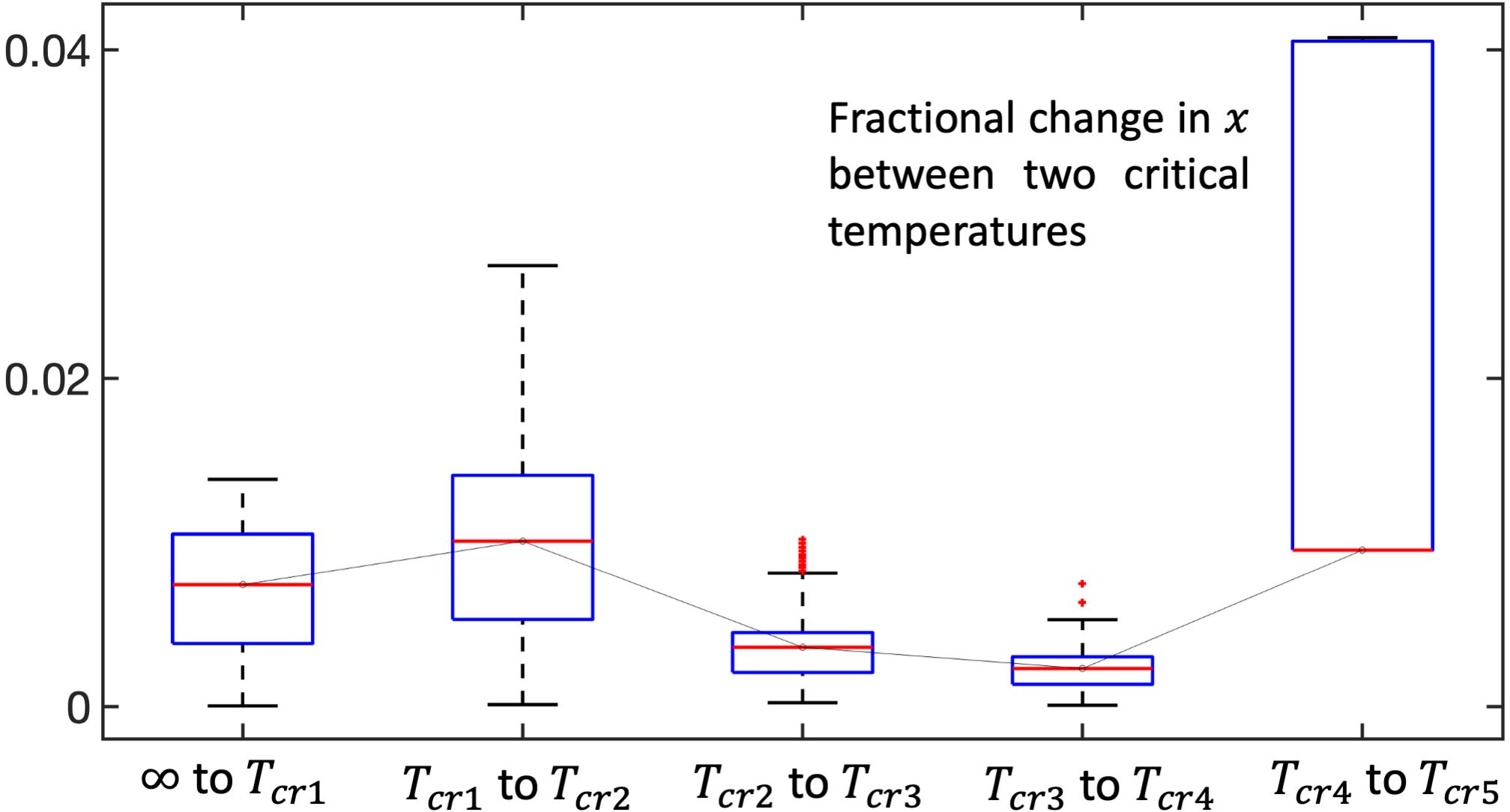}
    \caption{Plot of fractional change in $x$ in between two consecutive critical temperatures. The fractional change is computed as $\|\mathbf{x}(n) - \mathbf{x}_{\text{mean}}\|/\|\mathbf{x}_{\text{mean}}\|$, where $\mathbf{x}(n)$ denotes the $\mathbf{x}$ determined by Algorithm \ref{alg: algorithm_SLR} at the $n$-th iteration, and $\mathbf{x}_{\text{mean}}$ is the mean value of all the $\mathbf{x}(n)$'s occurring in between the corresponding two consecutive phase transitions.}
    \label{fig: PT_figure2}
\end{figure}
\begin{comment}Note that, upon substituting $x$ from (\ref{eq: x_SLR}) in the expression (\ref{eq: F_define}), the resulting Lagrangian $F$ is only a function of $Q$ and thus, the Hessian of $F$ is computed considering it only as function of $Q$ in (\ref{eq: Q_SLR}).
\end{comment}
The following theorem derives the critical temperature $T_{cr}$ using the criteria presented in (\ref{eq: HessSLR}).
\begin{theorem}\label{th: PT_SLR}
The value of critical temperature $T_{cr}$ where the Hessian $\mathcal{J}(Q,\Psi,T)$ in (\ref{eq: HessSLR}) becomes nonpositive for a feasible perturbation direction $\Psi$, triggering phase transition in the solution from Algorithm \ref{alg: algorithm_SLR}, is determined by:
\begin{align}
&T_{cr}:=2\max_{1\leq j\leq k}[T_{cr,j}], \text{where }T_{cr,j}=2\lambda_{\max}\big(D_{\Delta,j}\big). 
\end{align}
Here, $D_{\Delta,j}\in\mathbb{R}^{(d-1)\times (d-1)}$ is referred to as the PT (phase transition) matrix corresponding to the $j$-th column $q_j$ of $Q$ (or, in other words the $j$-th entry $x_j$ in $\mathbf{x}$). As described in the Appendix \ref{app: Thm_SLR}, the PT matrix $D_{\Delta,j}$ depends on $Q$ and a constant matrix of rank $d-1$.
\end{theorem}
\begin{proof}
Please refer to the Appendix \ref{app: Thm_SLR}
\end{proof}

{\em Annealing schedule: } Phase transition plays a key role in designing the annealing schedule for the temperature $T$ in the Algorithm \ref{alg: algorithm_SLR}. We observe that between two consecutive critical temperatures the change in the vector of non-zero values $\mathbf{x}$ as determined by the Algorithm \ref{alg: algorithm_SLR} is small. Let $\Delta \mathbf{x}(n) = \|\mathbf{x}(n)-\mathbf{x}_{\text{mean}}\|/\|\mathbf{x}_{\text{mean}}\|$ denote the fractional change in $\mathbf{x}(n)$ --- the non-zero values determined by the Algorithm \ref{alg: algorithm_SLR} at temperature $T_n$ --- where $\mathbf{x}_{\text{mean}}$ is the average of all the non-zero value vectors determined by the Algorithm \ref{alg: algorithm_SLR} in between the two consecutive critical temperatures $T_{cr_{n_1}}$ and $T_{cr_{n_2}}$ such that $T_{cr_{n_1}} < T_n < T_{cr_{n_2}}$. Figure \ref{fig: PT_figure2} illustrates the boxplot of this fractional change in $\Delta\mathbf{x}(n)$ observed in between two consecutive $T_{cr}'s$ for the example considered earlier in this section. Note that $\Delta \mathbf{x}(n)$ is quite small and the medain roughly lies between $0\%$ to $1\%$; whereas the change observed at the phase transition is considerable, as it adds a distinct non-zero value in $\mathbf{x}$ (see Figure \ref{fig: PT_figure}). In particular, in Figure \ref{fig: PT_figure2}, we observe a change of $84\%$ at the first critical temperature $T_{cr1}$, $32\%$ at $T_{cr2}$, $8\%$ at $T_{cr3}$ and $9\%$ at $T_{cr4}$ in the non-zero values $\mathbf{x}$ determined by the Algorithm \ref{alg: algorithm_SLR}. Thus, the solution given by the algorithm undergoes a drastic change only at $T_{cr}$'s and largely remains unchanged in between any two consecutive $T_{cr}$'s.

The above characteristic of our MEP-based algorithm is significant to determining an appropriate schedule for the annealing parameter $T$. In particular, one approach is to  solve the optimization problem (\ref{eq: MEP_sparse_Lin}) only at the critical temperatures $T_{cr}$, where the $T_{cr}$'s can be analytically computed as in Theorem \ref{th: PT_SLR}. Another heuristic approach, to avoid computing $T_{cr}$, is to geometrically anneal temperature $T$ as in other MEP-based frameworks \cite{rose1998deterministic}.

\begin{remark}
Phase transitions are also shown to provide  insights into the choice of hyper-parameters in various optimization problems. For instance, \cite{srivastava2019persistence} demonstrate its utility in determining the true number of clusters in a given dataset, and \cite{srivastava2022choice} use it in determining the number of states in an aggregated Markov chain. Similarly, phase transitions can be exploited to determine the best choice of the sparsity in SLR problems. The underlying idea is that if a given number distinct entries in $\mathbf{x}$  persist for a long range of change in $T$ values, then it is a good estimate of the (underlying) true sparsity. For instance, the true sparsity $\|w\|_0=3$ for the example considered in Figure \ref{fig: PT_figure}. It is evident from figure that $k_d=3$ distinct entries in $\mathbf{x}$ is observed for a longer range of temperature values in comparison to other $k_d$ values. Thus, estimating the true sparsity correctly. 
\end{remark}

\section{Simulations and Results}\label{sec: Simulations}

In this section, we apply our methodology to a dataset containing 205 data points related to automobile features, as referenced in \cite{misc_automobile_10}. Notably, 195 of these records have complete information, and we carefully select 13 continuous features for sparse linear regression, with the automobile price as the model's output. These features encompass a wide range of attributes, including car dimensions, weight, engine specifications, and fuel efficiency. To enhance the quality of results, we normalize the dataset in a way that ensures each column has a 2-norm of 1.

Our primary aim is to develop a predictive model for automobile prices, with an emphasis on sparsity. This involves the selection of a small-sized subset of these 13 features, accompanied by their respective coefficients. In essence, our goal is to determine the values of the matrix $V$ and vector $\mathbf{x}$ in (\ref{eq: Sparse_def2}) while maintaining a predefined sparsity level indicated by $k$. Here, $V$ represents the selected features, and $\mathbf{x}$ signifies their corresponding coefficients. For the unconstrained scenario, we consider three instances of sparsity $k\in\{3,4,5\}$. The resultant $V$ matrices are visually represented in Figure \ref{fig: c12} (a.1, a.2, and a.3). As illustrated in Section \ref{sec: PF_sparsity}, a feature $ a_j$ is selected if the $j$-th row of $V$ sums up to a value $\geq 1$. Thus, as illustrated in the Figure, the selected features correspond to columns $\{ a_6,  a_9,  a_{12}\}$, $\{ a_6,  a_9,  a_{10},  a_{12}\}$, and $\{ a_{6}, a_{8}, a_{9}, a_{10}, a_{12}\}$, respectively for the above three sparsity levels. The corresponding cost function values for these degrees of sparsity are $0.2248$, $0.2211$, and $0.2165$.

To assess our method's accuracy and establish a benchmark, we conduct a comparative analysis, evaluating our results against Trimmed Lasso, as presented in \cite{bertsimas2017trimmed}, using both the alternating minimization and convex envelopes heuristics. This analysis employs the same dataset and consistent sparsity levels, ensuring a fair comparison. Trimmed Lasso has consistently demonstrated superior performance compared to various other Lasso variants, making it an ideal point of comparison for our method.

The cost values associated with the solutions are graphically illustrated in Figure \ref{fig: Losses}. Notably, across this dataset and for three distinct sparsity levels ($k$ values of 3, 4, and 5), our maximum entropy approach and both Trimmed Lasso heuristics exhibit quite similar performance. It's worth mentioning that our method holds a slight advantage in terms of accuracy. The features selected by both heuristics of Trimmed Lasso correspond to columns $\{a_6, a_{10}, a_{12}\}$, $\{a_6, a_9, a_{10}, a_{12}\}$, and $\{a_6, a_8, a_9, a_{10}, a_{12}\}$ respectively, for the three mentioned sparsity levels.

\captionsetup[figure]{font=small,labelfont=small}
\begin{figure}[t]
\centering
  \includegraphics[width=0.5\textwidth]{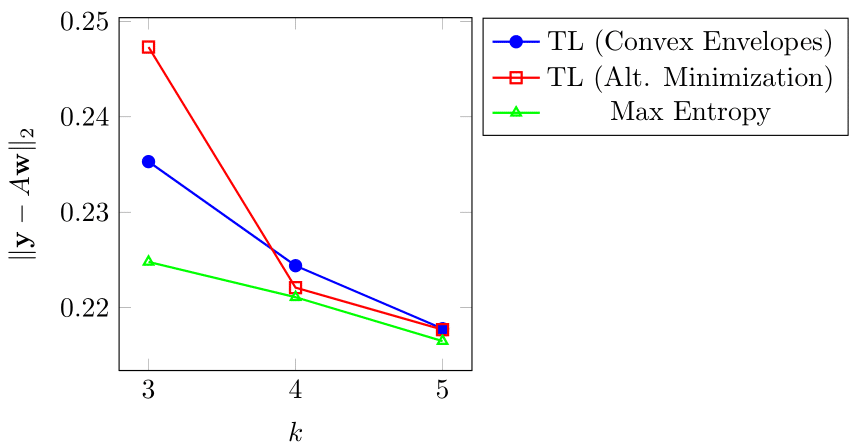}
  \caption{The figure illustrates the cost values associated with solutions obtained using the Alternating Minimization and Convex Envelopes heuristics of Trimmed Lasso, alongside our proposed Maximum-Entropy methodology.}
   \label{fig: Losses}
\end{figure}

A distinctive advantage of our approach, setting it apart from Trimmed Lasso, is its capacity to integrate diverse constraints, as discussed in Section \ref{sec: Flexibility}. In the rest of this section, we demonstrate the successful implementation of all described constraint types on the same dataset.

\captionsetup[figure]{font=small,labelfont=small}
\begin{figure*}[t]
\centering
  \includegraphics[width=0.58\textwidth]{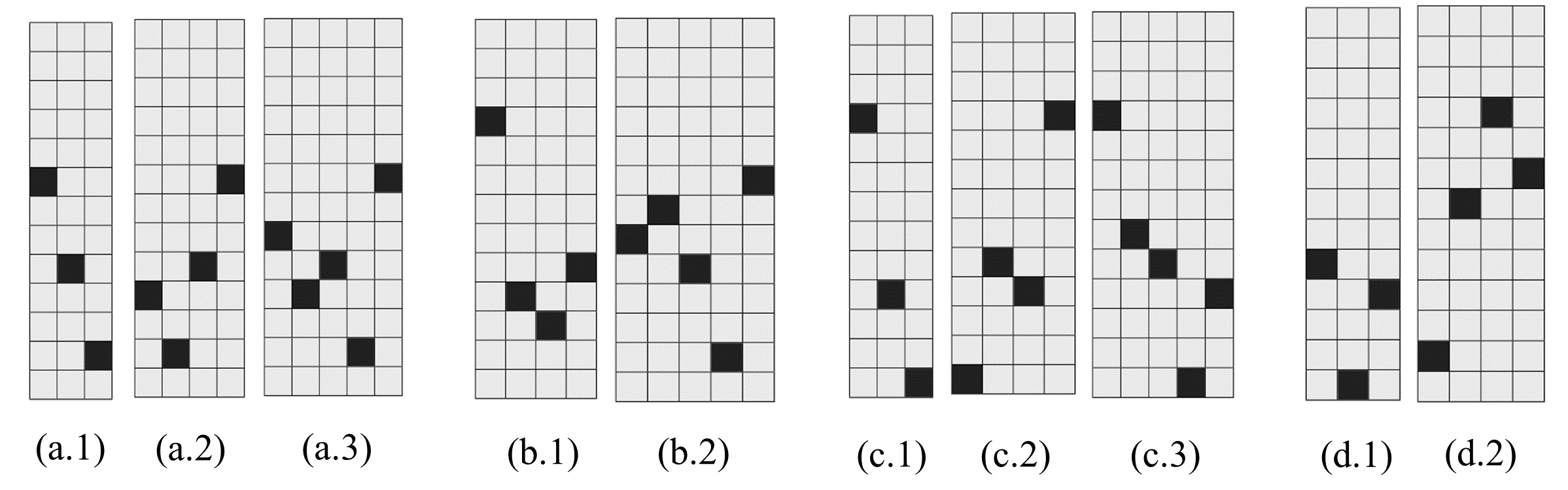}
  \caption{The figure exhibits the obtained $V$ matrices for different $k$ values (3, 4, and 5), encompassing various scenarios: unconstrained (a.1, a.2, a.3), correlated features removal (b.1, b.2), a priori knowledge imposition (c.1, c.2, c.3), and grouping constraints (d.1, d.2). In this representation, the color black denotes that the element has a value of one, while gray is used to indicate a value of zero.}
   \label{fig: c12}
\end{figure*}
\noindent
{\em 1- Correlated feature vectors: } We classify features as correlated if their absolute correlation coefficient exceeds 0.8. This criterion identifies the following sets of correlated features: $\{a_1, a_2, a_4\}$, $\{a_1, a_5\}$, $\{a_2, a_5\}$, $\{a_4, a_6\}$, $\{a_{10}, a_{12}, a_{13}\}$, $\{a_{10}, a_{6}\}$, and $\{a_{13}, a_4\}$.  As a result, it is advisable to choose, at most, one feature from each correlated set. As shown in Figure \ref{fig: c12}, for both $k=4$ and $k=5$ (a.2 and a.3), the unconstrained solution includes correlated features $\{a_6, a_{10}\}$ and $\{a_{10}, a_{12}\}$. Upon imposing the constraint (\ref{eq: Const1Correlation3}), the solution changes to selecting $\{a_4, a_9, a_{10}, a_{11}\}$ for $k=4$ and $\{a_6, a_7, a_8, a_9, a_{12}\}$ for $k=5$, as demonstrated in Figure \ref{fig: c12} (b.1 and b.2). The corresponding cost values are $0.2538$ and $0.2214$, respectively, which are (naturally) a bit larger than the unconstrained scenario illustrated above.

{\em 2- A priori knowledge: } Suppose we possess prior knowledge indicating that from various feature groups, we must include at least one feature. For example, in our dataset, it is essential to select at least one feature related to the vehicle's size, one related to its engine, and one associated with its fuel efficiency. In practical terms, this constraint implies that within each group of columns $\{a_1, a_2, a_3, a_4\}$, $\{a_6, a_7, a_8, a_9, a_{10}, a_{11}\}$, and $\{a_{12}, a_{13}\}$, a minimum of one column must be selected. The features selected under this constraint are depicted in Figure \ref{fig: c12} (c.1, c.2, and c.3), which correspond to $\{a_4, a_{10}, a_{13}\}$, $\{a_4, a_9, a_{10}, a_{13}\}$, and $\{a_4, a_8, a_9, a_{10}, a_{13}\}$ for sparsity levels of 3, 4, and 5, respectively. The respective cost values are $0.2657$, $0.2550$, and $0.2223$.

{\em 3- Grouping constraints:} Solely for demonstrative purposes, we have assumed that columns $\{a_6, a_7\}$ and $\{a_9, a_{10}\}$ are treated as unified groups. In other words, it's an all-or-nothing selection within each group. Illustrated in Figure \ref{fig: c12} (a.1 and a.2), the initial unconstrained solution contradicts the imposed constraint. After incorporating the grouping constraint, the newly selected features for $k=3$ and $k=4$ are $\{a_9, a_{10}, a_{13}\}$ and $\{a_4, a_6, a_7, a_{12}\}$, as depicted in Figure \ref{fig: c12} (d.1 and d.2) with cost values of $0.2655$ and $0.2268$, respectively.

\appendix
\section{Sparse Linear Regression}
\subsection{Algebraic Simplification}\label{app: Lin_obj_simplif}
\noindent
Algebraic simplification of the objective function (\ref{eq: MEP_sparse_obj})
\begin{equation*}
\begin{aligned}
    H &= -\sum_{V\in\mathcal{V}} p(V|\mathbf{x}) \log \; (p(V|\mathbf{x})) \\
    &= -\sum_{i,j = 1}^{d,k}\sum_{v_{ij} \in \{0,1\}}p_{ij}(v_{ij}|\mathbf{x})\log p_{ij}(v_{ij}|\mathbf{x}) \\ 
    &=-\sum_{i,j = 1}^{d,k}p_{ij}(v_{ij}=1|\mathbf{x})\log p_{ij}(v_{ij}=1|\mathbf{x}) \\
    &\qquad+p_{ij}(v_{ij}=0|\mathbf{x})\log p_{ij}(v_{ij}=0|\mathbf{x})\\
    & = - \; \mathbf{1}_d^{\top}[Q\circ\log Q+\bar{Q}\circ\log \bar{Q}]\mathbf{1}_k =: \mathcal{H}.
\end{aligned}
\end{equation*}
and constraint  (\ref{eq: MEP_sparse_c1})
\begin{equation}
    \begin{aligned}\label{eq: simplif_const}
        &D = \sum_{V\in\mathcal{V}}p(V|\mathbf{x})\|\mathbf{y}-A V\mathbf{x}\|_2^2 \\ 
        &=\sum_{V\in\mathcal{V}}\prod_{i,j=1}^{d,k}p_{ij}(v_{ij}|\mathbf{x})(\mathbf{y}-A V\mathbf{x})^{\top}(\mathbf{y}-A V\mathbf{x})\\
        &= \mathbf{y}^{\top}\mathbf{y} - 2\sum_{i,j=1}^{d,k}\mathbf{y}^{\top}  a_ix_j\sum_{V\in\mathcal{V}}p(V|\mathbf{x})v_{ij} + \\ 
        &\qquad \qquad \sum_{i,j,l,s=1}^{d,d,k,k}  a_i^{\top}  a_lx_jx_s\sum_{V\in\mathcal{V}} p(V|\mathbf{x})v_{ij}v_{ls}\\
        &= \mathbf{y}^{\top}\mathbf{y} - 2\sum_{i,j=1}^{d,k}y^{\top}   a_ix_j q_{ij} + \sum_{i,j=1}^{d,k}q_{ij}x_j{  a_i}^{\top} \sum_{l,s=1}^{d,k}q_{ls}x_s  a_l\\
        &\qquad \qquad
        + \sum_{i,j=1}^{d,k}(q_{ij}-q_{ij}^2)  a_k^{\top}  a_k x_i^2\\
        &=\|\mathbf{y}-A Q\mathbf{x}\|^2_2 + [  a_1^{\top} a_1~\cdots~ a_d^{\top} a_d][Q\circ{Q}](\mathbf{x}\circ\mathbf{x}) =: \mathcal{D}. \nonumber
    \end{aligned}
\end{equation}

\subsection{Proof of Theorem \ref{th: PT_SLR}}\label{app: Thm_SLR}
Note that $\Psi\in\mathbb{R}^{d\times k}$ is feasible perturbation if and only if $\mathbf{1}_d^{\top}\Psi=\mathbf{1}_k^{\top}$, i.e., its column sums up to $0$. In other words, its rank is $(d-1)$. Thus, we parameterize it as $\Psi=C\Phi$, where $\Phi=[\phi_1~\cdots~\phi_k]\in\mathbf{R}^{(d-1)\times k}$, and $C=\begin{bmatrix}\mathbf{I}_{(d-1)\times(d-1)}\\\mathbf{0}_{1\times (d-1)} \end{bmatrix} -\frac{1}{d}\mathbf{1}_{d\times(d-1)}$ is a rank $(d-1)$ matrix. The Hessian $ \mathcal{J}(Q,\Psi,T)$ in (\ref{eq: HessSLR}) is given by
\begin{equation}
\begin{aligned}
\mathcal{J}&= \underbrace{\sum_{j=1}^k\Phi_j^{\top}C^{\top}\Big[\frac{T}{2}\Upsilon_{j}-\Gamma_j\Big]C\Phi_j}_{I}+\underbrace{\Big\|\sum_{j=1}^k E_jC\Phi_j\Big\|^2}_{II} \nonumber
\end{aligned} 
\end{equation}
We now claim that $\mathcal{J}(Q,\Psi,T) > 0$ if and only if $I > 0$.\\
Note that the {\em if} direction is obvious since the second term II is non-negative. For the {\em only if} direction we will show that when $I =0$, we can determine a non-zero $\{{\Phi}_j\}$ such that $II=0$. Note that $I =0$ $\Rightarrow$ $T$ is such that 
\[
\Big|C^{\top}\frac{T\Upsilon_{j}}{2}C-C^{\top}\Gamma_jC\Big|=0
\]
Let $\bar{j}\in\{1,\hdots,k\}$ be chosen such that the above condition happens for largest value of $T$. In fact, there will be several coincident $x_j$ (equivalently, $q_j$) at this point to allow generation of distinct non-zero values in $\mathbf{x}$. Let $I_{\bar{j}}=\{j: q_j=q_{\bar{j}}, 1\leq j\leq k\}$. Let $\nu_{\bar{j}}$ be such that $\nu_{\bar{j}}^{\top}\Big(C^{\top}\frac{\Upsilon_{\bar{j}}}{2\beta}C-C^{\top}\Gamma_{\bar{j}}C\Big)\nu_{\bar{j}}=0$. Let
\[
\phi_{j}=
\begin{cases}
0,& \text{if }j\notin I_{\bar{j}}\\
\alpha_{j}\nu_{\bar{j}}, & \text{if } j\in I_{\bar{j}}
\end{cases}
\]
where $\sum_{j\in I_{\bar{j}}} \alpha_j = 0, \alpha_j\neq 0$.
In this case, 
\begin{align*}
\text{II} &=\Big\|\sum_j E_jC\alpha_j \nu_j\Big\|^2=\Big\|\sum_{j\in I_{\bar{j}}} E_jC\nu_j\alpha_j\Big\|^2\\
&=\Big\|E_{\bar{j}}C\nu_{\bar{j}}\sum_{j\in I_{\bar{j}}} \alpha_j\Big\|^2 = 0.
\end{align*}
Thus, $\mathcal{J}=0$ when $T$, $\nu_{j}$ satisfies $\nu_{j}^{\top}C^{\top}\Big[\frac{T}{2}\Upsilon_{b_{j}}-\Gamma_{j}\Big]C\nu_{j}=0$; this is possible only when 
\begin{align*}
\lambda_{\max}\Big[\frac{T}{2}C^{\top}\Upsilon_{j}C-C^{\top}\Gamma_{j}C\Big]=0.
\end{align*}
Note that $H_{0,j}\triangleq C^{\top}\Upsilon_{j}C>0$ and $H_{1,j} \triangleq C^{\top}\Gamma_{j}C=H_{1,j}^{\top}$. From Theorem 12.19 (Simultaneous Reduction to Diagonal Form) in \cite{laub2005matrix}, we have - there exists a non-singular matrix $G_j$ such that $G_j^{\top}H_{0,j}G_j = I$; $G_j^{\top}H_{1,j}G_j = D_{\Delta,j}$ where $D_{\Delta,j}$ is a diagonal matrix. Here $G_j =L_{\Delta,j}^{-\top}P_{\Delta,j}$ where $H_{0,j}=L_{\Delta,j}L_{\Delta,j}^{\top}$ (Choleski factorization of $H_{0,j}$) and $P_{\Delta,j}$ is such that $P_{\Delta,j}^{\top}[L_{\Delta,j}^{-1}H_{1,j}L_{\Delta,j}^{-\top}]P_{\Delta,j}=D_{\Delta,j}$ is a diagonal matrix. Such a $P_{\Delta,j}$ always exists since $H_{1,j}=H_{1,j}^{\top}$ [SVD of $L_{\Delta,j}^{-1}H_{1,j}L_{\Delta,j}^{-\top}$]. Let $\nu_{j} = G_j\bar{\nu}_{j}$
\begin{align*}
&\therefore \Rightarrow \bar{\nu}_{j}^{\top}\big[\frac{T}{2}G_j^{\top}H_{0,j}G_j - G_j^{\top}H_{1,j}G_j\big]\bar{\nu}_{j} = 0\\
&\Rightarrow \bar{\nu}_{j}^{\top}\big[\frac{T}{2}I - D_{\Delta,j}\big]\bar{\nu}_{j}= 0\\
&\Rightarrow \frac{T}{2}I - D_{\Delta,j} = 0 ~\Rightarrow T_{cr,j} = 2\lambda_{\max}\big(D_{\Delta,j}\big).
\end{align*}

\bibliographystyle{IEEEtran}
\bibliography{mybibfile}

\begin{thebibliography}{10}
\providecommand{\url}[1]{#1}
\csname url@samestyle\endcsname
\providecommand{\newblock}{\relax}
\providecommand{\bibinfo}[2]{#2}
\providecommand{\BIBentrySTDinterwordspacing}{\spaceskip=0pt\relax}
\providecommand{\BIBentryALTinterwordstretchfactor}{4}
\providecommand{\BIBentryALTinterwordspacing}{\spaceskip=\fontdimen2\font plus
\BIBentryALTinterwordstretchfactor\fontdimen3\font minus
  \fontdimen4\font\relax}
\providecommand{\BIBforeignlanguage}[2]{{%
\expandafter\ifx\csname l@#1\endcsname\relax
\typeout{** WARNING: IEEEtran.bst: No hyphenation pattern has been}%
\typeout{** loaded for the language `#1'. Using the pattern for}%
\typeout{** the default language instead.}%
\else
\language=\csname l@#1\endcsname
\fi
#2}}
\providecommand{\BIBdecl}{\relax}
\BIBdecl

\bibitem{wen2018survey}
F.~Wen, L.~Chu, P.~Liu, and R.~C. Qiu, ``A survey on nonconvex
  regularization-based sparse and low-rank recovery in signal processing,
  statistics, and machine learning,'' \emph{IEEE Access}, vol.~6, pp.
  69\,883--69\,906, 2018.

\bibitem{o2021sparse}
R.~J. O’Shea, S.~Tsoka, G.~J. Cook, and V.~Goh, ``Sparse regression in cancer
  genomics: comparing variable selection and predictions in real world data,''
  \emph{Cancer Informatics}, vol.~20, p. 11769351211056298, 2021.

\bibitem{fan2011sparse}
J.~Fan, J.~Lv, and L.~Qi, ``Sparse high-dimensional models in economics,''
  \emph{Annu. Rev. Econ.}, vol.~3, no.~1, pp. 291--317, 2011.

\bibitem{zhu2019scaled}
S.~Zhu and Y.~Wang, ``Scaled sequential threshold least-squares (s2tls)
  algorithm for sparse regression modeling and flight load prediction,''
  \emph{Aerospace Science and Technology}, vol.~85, pp. 514--528, 2019.

\bibitem{ChristosL0Sparse}
\BIBentryALTinterwordspacing
C.~Louizos, M.~Welling, and D.~P. Kingma, ``Learning sparse neural networks
  through l{\_}0 regularization,'' in \emph{6th International Conference on
  Learning Representations, {ICLR} 2018, Vancouver, BC, Canada, April 30 - May
  3, 2018, Conference Track Proceedings}.\hskip 1em plus 0.5em minus
  0.4em\relax OpenReview.net, 2018. [Online]. Available:
  \url{https://openreview.net/forum?id=H1Y8hhg0b}
\BIBentrySTDinterwordspacing

\bibitem{tuia2016nonconvex}
D.~Tuia, R.~Flamary, and M.~Barlaud, ``Nonconvex regularization in remote
  sensing,'' \emph{IEEE Transactions on Geoscience and Remote Sensing},
  vol.~54, no.~11, pp. 6470--6480, 2016.

\bibitem{welch1982algorithmic}
W.~J. Welch, ``Algorithmic complexity: three np-hard problems in computational
  statistics,'' \emph{Journal of Statistical Computation and Simulation},
  vol.~15, no.~1, pp. 17--25, 1982.

\bibitem{zhang2011sparse}
T.~Zhang, ``Sparse recovery with orthogonal matching pursuit under rip,''
  \emph{IEEE transactions on information theory}, vol.~57, no.~9, pp.
  6215--6221, 2011.

\bibitem{donoho2012sparse}
D.~L. Donoho, Y.~Tsaig, I.~Drori, and J.-L. Starck, ``Sparse solution of
  underdetermined systems of linear equations by stagewise orthogonal matching
  pursuit,'' \emph{IEEE transactions on Information Theory}, vol.~58, no.~2,
  pp. 1094--1121, 2012.

\bibitem{zhang2008adaptive}
T.~Zhang, ``Adaptive forward-backward greedy algorithm for sparse learning with
  linear models,'' \emph{Advances in neural information processing systems},
  vol.~21, 2008.

\bibitem{shekhar2022forward}
P.~Shekhar and A.~Patra, ``A forward--backward greedy approach for sparse
  multiscale learning,'' \emph{Computer Methods in Applied Mechanics and
  Engineering}, vol. 400, p. 115420, 2022.

\bibitem{zhang2015survey}
Z.~Zhang, Y.~Xu, J.~Yang, X.~Li, and D.~Zhang, ``A survey of sparse
  representation: algorithms and applications,'' \emph{IEEE access}, vol.~3,
  pp. 490--530, 2015.

\bibitem{figueiredo2007gradient}
M.~A. Figueiredo, R.~D. Nowak, and S.~J. Wright, ``Gradient projection for
  sparse reconstruction: Application to compressed sensing and other inverse
  problems,'' \emph{IEEE Journal of selected topics in signal processing},
  vol.~1, no.~4, pp. 586--597, 2007.

\bibitem{beck2009fast}
A.~Beck and M.~Teboulle, ``A fast iterative shrinkage-thresholding algorithm
  for linear inverse problems,'' \emph{SIAM journal on imaging sciences},
  vol.~2, no.~1, pp. 183--202, 2009.

\bibitem{yang2013linearized}
J.~Yang and X.~Yuan, ``Linearized augmented lagrangian and alternating
  direction methods for nuclear norm minimization,'' \emph{Mathematics of
  computation}, vol.~82, no. 281, pp. 301--329, 2013.

\bibitem{davis1997adaptive}
G.~Davis, S.~Mallat, and M.~Avellaneda, ``Adaptive greedy approximations,''
  \emph{Constructive approximation}, vol.~13, pp. 57--98, 1997.

\bibitem{LARS}
\BIBentryALTinterwordspacing
B.~Efron, T.~Hastie, I.~Johnstone, and R.~Tibshirani, ``Least angle
  regression,'' \emph{The Annals of Statistics}, vol.~32, no.~2, pp. 407--451,
  2004. [Online]. Available: \url{http://www.jstor.org/stable/3448465}
\BIBentrySTDinterwordspacing

\bibitem{amir2021trimmed}
T.~Amir, R.~Basri, and B.~Nadler, ``The trimmed lasso: Sparse recovery
  guarantees and practical optimization by the generalized soft-min penalty,''
  \emph{SIAM journal on mathematics of data science}, vol.~3, no.~3, pp.
  900--929, 2021.

\bibitem{wang2021non}
J.~Wang, ``Non-convex lp regularization for sparse reconstruction of electrical
  impedance tomography,'' \emph{Inverse Problems in Science and Engineering},
  vol.~29, no.~7, pp. 1032--1053, 2021.

\bibitem{9585422}
P.~K. Pokala, R.~V. Hemadri, and C.~S. Seelamantula, ``Iteratively reweighted
  minimax-concave penalty minimization for accurate low-rank plus sparse matrix
  decomposition,'' \emph{IEEE Transactions on Pattern Analysis and Machine
  Intelligence}, vol.~44, no.~12, pp. 8992--9010, 2022.

\bibitem{SCAD}
\BIBentryALTinterwordspacing
Y.~Kim, H.~Choi, and H.-S. Oh, ``Smoothly clipped absolute deviation on high
  dimensions,'' \emph{Journal of the American Statistical Association}, vol.
  103, no. 484, pp. 1665--1673, 2008. [Online]. Available:
  \url{http://www.jstor.org/stable/27640214}
\BIBentrySTDinterwordspacing

\bibitem{bertsimas2017trimmed}
D.~Bertsimas, M.~S. Copenhaver, and R.~Mazumder, ``The trimmed lasso: Sparsity
  and robustness,'' \emph{arXiv preprint arXiv:1708.04527}, 2017.

\bibitem{wang2008note}
H.~Wang and C.~Leng, ``A note on adaptive group lasso,'' \emph{Computational
  statistics \& data analysis}, vol.~52, no.~12, pp. 5277--5286, 2008.

\bibitem{yuan2011efficient}
L.~Yuan, J.~Liu, and J.~Ye, ``Efficient methods for overlapping group lasso,''
  \emph{Advances in neural information processing systems}, vol.~24, 2011.

\bibitem{micchelli2013regularizers}
C.~A. Micchelli, J.~M. Morales, and M.~Pontil, ``Regularizers for structured
  sparsity,'' \emph{Advances in Computational Mathematics}, vol.~38, pp.
  455--489, 2013.

\bibitem{baldassarre2012incorporating}
L.~Baldassarre, J.~M. Morales, and M.~Pontil, ``Incorporating additional
  constraints in sparse estimation,'' \emph{IFAC Proceedings Volumes}, vol.~45,
  no.~16, pp. 959--964, 2012.

\bibitem{jacob2009group}
L.~Jacob, G.~Obozinski, and J.-P. Vert, ``Group lasso with overlap and graph
  lasso,'' in \emph{Proceedings of the 26th annual international conference on
  machine learning}, 2009, pp. 433--440.

\bibitem{rose1998deterministic}
K.~Rose, ``Deterministic annealing for clustering, compression, classification,
  regression, and related optimization problems,'' \emph{Proceedings of the
  IEEE}, vol.~86, no.~11, pp. 2210--2239, 1998.

\bibitem{xu2014aggregation}
Y.~Xu, S.~M. Salapaka, and C.~L. Beck, ``Aggregation of graph models and markov
  chains by deterministic annealing,'' \emph{IEEE Transactions on Automatic
  Control}, vol.~59, no.~10, pp. 2807--2812, 2014.

\bibitem{baranwal2022unified}
M.~Baranwal, L.~Marla, C.~Beck, and S.~M. Salapaka, ``A unified maximum entropy
  principle approach for a large class of routing problems,'' \emph{Computers
  \& Industrial Engineering}, vol. 171, p. 108383, 2022.

\bibitem{9517030}
A.~Srivastava and S.~M. Salapaka, ``Parameterized mdps and reinforcement
  learning problems--a maximum entropy principle-based framework,'' \emph{IEEE
  Transactions on Cybernetics}, 2021.

\bibitem{chen2005protein}
L.~Chen, T.~Zhou, and Y.~Tang, ``Protein structure alignment by deterministic
  annealing,'' \emph{Bioinformatics}, vol.~21, no.~1, pp. 51--62, 2005.

\bibitem{yu2013maximal}
J.-G. Yu, J.~Zhao, J.~Tian, and Y.~Tan, ``Maximal entropy random walk for
  region-based visual saliency,'' \emph{IEEE transactions on cybernetics},
  vol.~44, no.~9, pp. 1661--1672, 2013.

\bibitem{jaynes2003probability}
E.~T. Jaynes, \emph{Probability theory: The logic of science}.\hskip 1em plus
  0.5em minus 0.4em\relax Cambridge university press, 2003.

\bibitem{jaynes1957information}
------, ``Information theory and statistical mechanics,'' \emph{Physical
  review}, vol. 106, no.~4, p. 620, 1957.

\bibitem{rose1991deterministic}
K.~Rose, ``Deterministic annealing, clustering, and optimization,'' Ph.D.
  dissertation, California Institute of Technology, 1991.

\bibitem{bertsekas2014constrained}
D.~P. Bertsekas, \emph{Constrained optimization and Lagrange multiplier
  methods}.\hskip 1em plus 0.5em minus 0.4em\relax Academic press, 2014.

\bibitem{hestenes1969multiplier}
M.~R. Hestenes, ``Multiplier and gradient methods,'' \emph{Journal of
  optimization theory and applications}, vol.~4, no.~5, pp. 303--320, 1969.

\bibitem{powell1969method}
M.~J. Powell, ``A method for nonlinear constraints in minimization problems,''
  \emph{Optimization}, pp. 283--298, 1969.

\bibitem{cohen2009pearson}
I.~Cohen, Y.~Huang, J.~Chen, J.~Benesty, J.~Benesty, J.~Chen, Y.~Huang, and
  I.~Cohen, ``Pearson correlation coefficient,'' \emph{Noise reduction in
  speech processing}, pp. 1--4, 2009.

\bibitem{byrd1999interior}
R.~H. Byrd, M.~E. Hribar, and J.~Nocedal, ``An interior point algorithm for
  large-scale nonlinear programming,'' \emph{SIAM Journal on Optimization},
  vol.~9, no.~4, pp. 877--900, 1999.

\bibitem{coleman1996interior}
T.~F. Coleman and Y.~Li, ``An interior trust region approach for nonlinear
  minimization subject to bounds,'' \emph{SIAM Journal on optimization},
  vol.~6, no.~2, pp. 418--445, 1996.

\bibitem{srivastava2020simultaneous}
A.~Srivastava and S.~M. Salapaka, ``Simultaneous facility location and path
  optimization in static and dynamic networks,'' \emph{IEEE Transactions on
  Control of Network Systems}, 2020.

\bibitem{srivastava2019persistence}
A.~Srivastava, M.~Baranwal, and S.~Salapaka, ``On the persistence of clustering
  solutions and true number of clusters in a dataset,'' in \emph{Proceedings of
  the AAAI Conference on Artificial Intelligence}, vol.~33, 2019, pp.
  5000--5007.

\bibitem{srivastava2022choice}
A.~Srivastava, R.~K. Velicheti, and S.~M. Salapaka, ``On the choice of number
  of superstates in the aggregation of markov chains,'' \emph{Pattern
  Recognition Letters}, vol. 159, pp. 181--188, 2022.

\bibitem{misc_automobile_10}
J.~Schlimmer, ``{Automobile},'' UCI Machine Learning Repository, 1987, {DOI}:
  https://doi.org/10.24432/C5B01C.

\bibitem{laub2005matrix}
A.~J. Laub, \emph{Matrix analysis for scientists and engineers}.\hskip 1em plus
  0.5em minus 0.4em\relax Siam, 2005, vol.~91.

\end{thebibliography}

\end{document}